\documentclass{article}
\usepackage{arxiv}
\usepackage[utf8]{inputenc} % allow utf-8 input
\usepackage[T1]{fontenc}    % use 8-bit T1 fonts
\usepackage{hyperref}       % hyperlinks
\usepackage{url}            % simple URL typesetting
\usepackage{booktabs}       % professional-quality tables
\usepackage{amsfonts}       % blackboard math symbols
\usepackage{nicefrac}       % compact symbols for 1/2, etc.
\usepackage{microtype}      % microtypography
\usepackage{lipsum}
\usepackage{amsthm}
\usepackage{graphicx}
\usepackage{authblk}

\usepackage{comment}
\usepackage{multirow}
\usepackage[ruled,norelsize]{algorithm2e}

\SetAlFnt{\small}
\SetAlCapFnt{\small}
\SetAlCapNameFnt{\small}
\SetAlCapHSkip{0pt}
\IncMargin{-\parindent}
\usepackage{algorithmic}
\usepackage{wrapfig}

\newtheorem{theorem}{Theorem}

\newtheorem{lemma}{Lemma}

\newtheorem{assumption}{Assumption}[section]

\usepackage{mathrsfs}
\usepackage{color}
\usepackage{amsmath}
\usepackage{nomencl}
\makenomenclature

\usepackage{etoolbox}
\renewcommand\nomgroup[1]{%
  \item[\bfseries
  \ifstrequal{#1}{F}{Functions}{%
  \ifstrequal{#1}{V}{Variables}{%
  \ifstrequal{#1}{S}{Sets}{%
  \ifstrequal{#1}{E}{Empirical Study}}}}%
]}
\usepackage{amsbsy}    % math symbols
\usepackage{bbm}
\usepackage[numbers]{natbib}
\usepackage{color}
\usepackage{soul}
\begin{document}

\title{A Pooled Quantile Estimator for Parallel Simulations}

%\author{
%  Qiong~Zhang\\
%   Department of Computer Science\\
%  Clemson University\\
%  Clemson, SC 29634 \\
%  \texttt{qiongz@clemson.edu} \\
  %% examples of more authors
%  \And
%  Bo Wang\\
%   Department of Electrical Engineering\\
%  Northeastern University\\
%  Boston, MA 02115 \\
%  \texttt{wang.bo2@husky.neu.edu} \\
%  \AND
%  Wei Xie\\
%   Department of Electrical Engineering\\
%  Northeastern University\\
%  Boston, MA 02115 \\
%  \texttt{w.xie@northeastern.edu} \\
  %% \And
  %% Coauthor \\
  %% Affiliation \\
  %% Address \\
  %% \texttt{email} \\
  %% \And
  %% Coauthor \\
  %% Affiliation \\
  %% Address \\
  %% \texttt{email} \\
%}
\author[1]{Qiong Zhang}
\author[2]{Bo Wang}
\author[ 2]{Wei Xie\thanks{Corresponding author: w.xie@northeastern.edu}
}

\affil[1]{Clemson University, Clemson, SC 29634}
\affil[2]{Northeastern University, Boston, MA 02115}

% Page heads
%\markboth{Y. Yi and W. Xie}{A Metamodel-Assisted Framework for Two-Stage Optimization via Simulation}

%{\textcolor{red}{Global-Local Metamodel Assisted Two-Stage Optimization via Simulation}}

% \author[1]{Wei Xie}
% \author[2]{Yuan Yi}
% \author[3]{Hua Zheng}
% \affil[1]{Northeastern University}
% \affil[2]{Rensselaer Polytechnic Institute}
% \affil[3]{Northeastern University}

% \renewcommand\Authands{ and }
% \author{}
% \affiliation{%
% 	\institution{Northeastern University}
% 	\city{Boston}
% 	\state{MA}
% 	\postcode{02115}
% 	\country{USA}}
% \email{w.xie@northeastern.edu}

% \author{Yuan Yi}
% %\orcid{0000-0002-8184-2079}
% \affiliation{%
% 	\institution{Rensselaer Polytechnic Institute}
% 	\streetaddress{110 8th ST}
% 	\city{Troy}
% 	\state{NY}
% 	\postcode{12180}
% 	\country{USA}}
% \email{yiy2@rpi.edu}

% \author{Hua Zheng}
% \affiliation{%
% 	\institution{Northeastern University}
% 	\city{Boston}
% 	\state{MA}
% 	\postcode{02115}
% 	\country{USA}}
% \email{hua.zheng0908@gmail.com}

% NOTE! Affiliations placed here should be for the institution where the
%       BULK of the research was done. If the author has gone to a new
%       institution, before publication, the (above) affiliation should NOT be changed.
%       The authors 'current' address may be given in the "Author's addresses:" block (below).
%       So for example, Mr. Abdelzaher, the bulk of the research was done at UIUC, and he is
%       currently affiliated with NASA.
\maketitle
\begin{abstract}

	Quantile is an important risk measure quantifying the stochastic system random behaviors. This paper studies a pooled quantile estimator, which is the sample quantile of detailed simulation outputs after directly pooling independent sample paths together.
	We derive the
	asymptotic representation of the pooled quantile estimator and further prove its normality. By comparing with the classical quantile estimator used in stochastic simulation, both theoretical and empirical studies demonstrate the advantages of the proposal under the context of parallel simulation. 

\end{abstract}


% At a minimum you need to supply the author names, year and a title.
% IMPORTANT:
% Full first names whenever they are known, surname last, followed by a period.
% In the case of two authors, 'and' is placed between them.
% In the case of three or more authors, the serial comma is used, that is, all author names
% except the last one but including the penultimate author's name are followed by a comma,
% and then 'and' is placed before the final author's name.
% If only first and middle initials are known, then each initial
% is followed by a period and they are separated by a space.
% The remaining information (journal title, volume, article number, date, etc.) is 'auto-generated'.

%\begin{bottomstuff}
%Author's addresses: Y. Yi, email:~yiy2@rpi.edu; Department of Industrial and Systems Engineering,
%Rensselaer Polytechnic Institute, Troy, NY 12180-3590;
%W. Xie (corresponding author), email:~xiew3@rpi.edu; Department of Industrial and Systems Engineering,
%Rensselaer Polytechnic Institute, Troy, NY 12180-3590;  

%\end{bottomstuff}

%\renewcommand{\shortauthors}{Zhang, Wang and Xie}

\section{Introduction}
\label{sec:introduction}

Discrete-event simulation is often used to assess the performance of complex stochastic systems, especially in the situations where the direct analytical solution and physical experiments are infeasible or prohibitive \citep{banks2010discrete}. Advanced computer architectures have made parallel computing available and popular in many engineering and scientific areas. Nelson in \cite{Nelson_2016} 
raises new research questions about how to  {exploit} this computing advantage in estimating simulation system performance, especially risk measures, such as quantiles.
As mentioned in \cite{Nelson_2016}, a fundamental challenge is how to efficiently utilize all available parallel computing processors to improve the estimation accuracy.

%\textbf{give examples of application areas, and cite a review paper, do we need to limit to steady state simulation?}. 
In this paper, we consider the steady-state system performance. A single run of simulation generates a sample path of detailed outputs with a given run-length.
For various existing system performance (e.g., mean and risk) estimation approaches, both variance and bias of their estimators can be reduced 
by increasing the run-length \citep{Nelson_2016}. Since the detailed outputs in a simulation sample path {are} generated sequentially, it could be challenging to chop a dependent sample path into chunks and {run} separately in parallel from multiple processors. As a result, the run-length becomes the key bottleneck in improving the computational efficiency with parallel simulation.

For the classical quantile estimation approach, we typically calculate 
the sample quantile of outputs from each simulation run and then take the average of
the quantile estimators from multiple replications.
%, which we refer to as \textit{the average quantile estimator} in this paper. 
The asymptotic properties of this estimator have been well-studied in the literature (e.g., \cite{sun2010asymptotic}).
%Although with solid theoretical function, further improving the estimation accuracy of tail quantile measures is still an open topic for simulation analysis.
As noted in \cite{heidelberger1984quantile},
%\textbf{add some literature}, 
accurate estimators of tail quantile measures greatly rely on 
a sufficiently large run-length, which could be a luxury for complex and fast-evolving stochastic systems.
For example, we are interested in the 95\% quantile of waiting time in a queueing service system. 
%(\textbf{queueing system does not always have theoretical quantiles, right?}) \textbf{Raise a situation where there will be deadline to submit computing job},
Since the system is required to adapt fast to the evolving demand,
we need to assess the system performance and make decisions under a certain tight time deadline. 
Notice that the run-length greatly relates to the simulation running time before the deadline. 
Therefore, the parallel processors can increase the number of replication and utilize the available parallel processor
under an urgent deadline. However, the classical quantile estimators may not make efficient use of the detailed output sample paths. This could impact the quantile estimation accuracy, especially when the time budget is tight. 
%Under this scenario, parallel computing appears to be less powerful in simulation analysis. 

%As noted earlier, 
%most sample path-based quantile estimators take average over 
%the individual quantile estimators from different sample paths. 

	In this paper, we introduce a quantile estimator which is computed by directly pooling the detailed simulation outputs from various replications to obtain a sample quantile estimator,
	called \emph{pooled sample quantile estimator}. By pooling the dependent (within each replication) and independent (cross different replications) simulation outputs, the resulted sample quantile is introduced to estimate the quantile.
	The pooled quantile estimator has been investigated under independent and identically distributed observations, such as \cite{asmussen2007stochastic, nakayama2014confidence}, and recently been used to construct the confidence interval of quantile estimators in \cite{alexopoulos2019sequest}.
	%we develop asymptotic representation of the proposed pooled sample quantile estimator under the $\phi-$mixing assumption; see for example \cite{sen1972bahadur}.
	In this paper, we develop the asymptotic results of the pooled quantile estimator based on the framework in \cite{sen1972bahadur}. Our asymptotic results show that the proposed estimator has better performance especially under the situations with multiple processors and urgent time deadlines, i.e., the job request is too urgent to produce sufficiently longer sample paths. We highlight our contribution
	as follows.
	\begin{itemize}
		\item We provide the asymptotic results of the pooled quantile estimator that are generated from multiple replications of dependent sequences. 
		\item We illustrate how the pooled quantile estimator can be used to improve the accuracy of the classical average quantile estimator under the context of parallel simulation.
		%\item The proposed pooled quantile estimator is simple. Various variance reduction techniques in the literature could be potentially combined with our proposal to further improve the estimation accuracy of risk measures.
	\end{itemize}

% =============== Section 2 Problem Statement =====================================

\section{A pooled quantile estimator for parallel simulation}
\label{sec:problem_statement}

Our goal is to estimate the quantile of the marginal distribution for steady-state simulation. %The quantile estimation will directly enable statistical inference of the risk measures such as Value at Risk and Conditional Value at Risk. 
Let $X$ be a random variable representing a single entry in a
detailed simulation output sample path.
We denote the marginal cumulative distribution function (CDF) by $F(x) = \mbox{P}(X\le x)$, and denote the $\alpha$-level quantile by $\xi_\alpha \equiv F^{-1}(\alpha)$, where $0<\alpha<1$. 

We consider the detailed outputs generated from the steady-state stochastic simulation model,
\begin{equation}\label{eq:pool}
	\{X_{ji}; j=1,2,\cdots,R; i=1,2,\cdots,L\}
\end{equation}
with $R$ independent sample paths each with run-length $L$. Particularly, $X_{ji}$ denotes the $i$-th element of the $j$-th sample path output. Notice that different sample paths are independent with each other, but entries within each sample path are element-wise dependent. 
%\textbf{Just explain that the dependence are from two entries from a same sample path, but independent over different sample paths. We consider the common case that we run simulation model independently to produce each sample path, indicating output data are independent across different replications, i.e. $X_{ji}$ and $X_{j^\prime i^\prime}$ for $j \neq j^\prime$ are independent with each other. }
The pooled outputs in \eqref{eq:pool} contain 
$N=L\times R$ entries, which construct an empirical CDF:
\begin{equation}
	F_N(x) = \dfrac{1}{N}\sum_{j=1}^{R}\sum_{i=1}^{L}\mathbb{I}(X_{ji}\le x)
	, \label{eq:empcdf}
\end{equation}
where $\mathbb{I}(\cdot)$ is the indicator function. Let $X_{(1)}\le \cdots \le X_{(N)}$ denote the order statistics over all the entries in \eqref{eq:pool}.
We define the $\alpha$-level pooled quantile estimator obtained by combining all the entries by
%(\textbf{let's change the notation to be $\hat{\xi}^{(P)}_{p}$})
\begin{equation} \label{eq:qhat_ad}
	\hat{\xi}^{(P)}_{\alpha} = X_{(\lceil N\alpha\rceil)},
\end{equation}
where $\lceil a\rceil$ denotes the smallest integer greater than or equal to $a$.

As noted earlier in Section~\ref{sec:introduction}, the classical quantile estimator (see \cite{chen2006quantile, bekki2009indirect} for examples) is to obtain the $\alpha$-level sample quantile from each replication, $\hat{\xi}_{j, \alpha} = X_{j,(\lceil L\alpha\rceil)}$ with $X_{j,(1)} \le \cdots \le X_{j,(L)}$ for $j=1, \ldots, R$. The final estimator is 
\begin{equation}
	\hat{\xi}^{(A)}_{\alpha} = \dfrac{1}{R}\sum_{j=1}^{R}\hat{\xi}_{j, \alpha},
	\label{eq:qhat_cl}
\end{equation}
which is referred to as the \textit{average quantile estimator} in this paper. The asymptotic properties of each individual $\hat{\xi}_{j, \alpha}$ have been studied by a vast collection of papers; see for example \cite{sen1972bahadur}. Thus, the corresponding asymptotic properties of the average quantile estimator can be directly developed based on the mutual independence among $\hat{\xi}_{j, \alpha}$ for $j=1,2,\ldots,R$. 

{The pooled quantile estimator has been investigated under independent and identically distributed (i.i.d.) observations, such as \cite{asmussen2007stochastic, nakayama2014confidence}.
	In \cite{alexopoulos2019sequest}, the comparison of the asymptotic properties between the pooled quantile estimator and the average quantile estimator based on i.i.d. observations have been stated. Also, the pooled estimator from dependent sequences is used to adjust the confidence interval of quantile estimators. However, the comparison of asymptotic  properties between the pooled quantile estimator and the average quantile estimator based on multiple replications of dependent sequences has not been formally stated in the literature. }
%There is a vast collection of literatures on the quantile estimation and their corresponding asymptotic properties. \cite{bahadur1966note} established the famous Bahadur representation and asymptotic distribution of sample quantile from independent identically distributed (i.i.d.) samples.
%Asymptotic representations for quantile estimates have been also extended to dependent data, with various structural assumptions, e.g. $\phi-$mixing \cite{sen1972bahadur}, $\beta-$mixing \cite{drees2003extreme}, and $\alpha-$mixing \cite{yoshihara1995bahadur}
%However, these papers focus on observational data with only one dependent sequence, which can not directly apply to the pooled quantile estimator constructed by stochastic simulation outputs with multiple sample paths. Thus, we develop asymptotic representation to fill this gap. 
%Different from the average quantile estimator, it is non-trivial to study the classical asymptotic properties of the pooled quantile estimator. 
To fill this gap, we provide the theoretical comparison of the pooled quantile estimator and the average quantile estimator. 
Our theory is developed based on Assumptions 2.1--2.4.
\begin{assumption}\label{assumption1}
	For each replication $j$ ($=1,\cdots,R$), $\{X_{ji}; -\infty<i<\infty\}$ is a stationary sequence of $\phi$-mixing random variables, i.e., for the $\sigma-$fields $\mathscr{F}_{-\infty}^{k}$ and $\mathscr{F}_{k+n}^{\infty}$ generated by $\{X_{ji}; i\le k\}$ and $\{X_{ji}; i\ge k+n\}$ respectively, we have
	\begin{equation}
		\left|\mbox{P}(E_2\mid E_1)-\mbox{P}(E_2)\right|\le \phi(n),~~
		\mathrm{for}~~-\infty<k<\infty ~~\mathrm{and}~~n \ge 1\hfill \label{eq:mixing1}
	\end{equation}
	where $E_1 \in \mathscr{F}_{-\infty}^{k}$ and $E_2 \in \mathscr{F}_{k+n}^{\infty}$, and 
	$1\ge\phi(1)\ge\phi(2)\ge\cdots\geq 0$ with $\lim_{n\to\infty}\phi(n)=0$, and
	\begin{equation}
		\sum_{n=1}^{\infty}e^{tn}\phi(n)<\infty\ \mathrm{for}\ \mathrm{some}\ t>0.
		\label{eq:mixing2}
	\end{equation}
\end{assumption}

\begin{assumption}\label{assumption2}
	$R=o(L)$ as $L\rightarrow \infty$.
\end{assumption}

\begin{assumption}\label{assumption3}
	$F^\prime(x)=f(x)$ is continuous and  {positive} in the neighborhood of $\xi_\alpha$.
\end{assumption}

\begin{assumption}\label{assumption4}
	$f^\prime(x)=\dfrac{d}{dx}f(x)$ is positive and bounded in the neighborhood of $\xi_\alpha$.
\end{assumption}

Assumption~\ref{assumption1} is called the $\phi$-mixing condition, which is commonly 
adopted in the steady-state simulation output analysis \citep{chen2000batching, steiger2001convergence}. It states that the serial dependency decreases as the lag increases. 
The studies in {\cite{bradley2005basic} and \cite{bradley1986basic}  provided the results of  theoretically verifying the  $\phi$-mixing condition for some popular examples of dependent sequences, such as Markov chains, stationary Gaussian processes, and etc. }
Assumption~\ref{assumption2} is that we normally consider the run-length far larger than the number of replications, and it
also matches the situation in parallel computing where the number of
available processors is often less than the run-length for the steady-state simulation. Assumptions \ref{assumption3} and \ref{assumption4} are the common assumptions for developing asymptotic representations of quantile estimators.

%Assumption~\ref{assumption2} is used to develop asymptotic characteristics of pooled quantile estimator. \textbf{first say that we normally consider runlength far more larger than 
%number of replications, then mention that it also matches the situation in parallel computing}. It is also reasonable under parallel computing since the number of available processors is usually far more smaller than the the run-length. 

%(\textbf{express the $\phi$-mixing condition as Assumption 1. Also express the assumption $R=o(L)$ as
%assumption 2. Then briefly explain why we need these two assumptions. })

%In the follow part of this paper, we would first discuss the asymptotic behavior of sample quantile \eqref{eq:qhat_ad} and compare its properties with the classical quantile estimate given by \eqref{eq:qhat_cl}. 

% ============== Section 3 problem description ==============================

\section{The asymptotic representation of the pooled quantile estimator}
\label{sec:main_results}

\begin{sloppypar}
	The asymptotic representation of sample quantile has been investigated under both independent data and dependent sequence data. For the pooled quantile estimator of simulation outputs from independent replications, asymptotic characterization is still missing in the literature. We aim to fill this gap and develop the asymptotic representation for the quantile estimator with pooled sample paths and provide the theoretical insights on how to deploy multiple processors to improve the estimation of system quantile response. We first provide the asymptotic representation of the pooled quantile estimator through Theorem~\ref{theorem1}, and then Theorem~\ref{theorem3} gives its asymptotic distribution. The proofs of these Theorems follows a similar logic as in  \cite{sen1972bahadur},
	and extend their results to incorporate multiple replications of dependent sample paths. 
	%, but the results in  \cite{sen1972bahadur}  not directly imply our theoretical results. 
	%Our theories 
\end{sloppypar}

%(\textbf{Need to state that asymptotic representation has been investigated under what kind of situation, and why we need new theoretical results for our situation. What is the major difficulty in extending existing theoretical results to our situation. Briefly state how do you handle this. 
%Also need to add some description to link these theoretical results (basically, what is the intuition from each results. perhaps these description can be linked to parallel simulation). })

%Now we can provide the main results, let simulation output data $\{X_{ji}\}$ as described in Section~\ref{sec:problem_statement}, empirical CDF $F_N(x)$ defined as \eqref{eq:empcdf}, and sample $p-$quantile $\hat{\xi}^{(P)}_{p}$ given as \eqref{eq:qhat_ad}. 

\begin{theorem}\label{theorem1}
	Consider a small neighborhood around the true $\alpha$-level quantile $\xi_\alpha$, denoted by
	$I_N = \left\{x: \big|x - \xi_\alpha \big| \le N^{-1/2}\log L \right\}$.
	Under Assumptions 2.1--2.3, as $L\to \infty$,
	\begin{equation}
		\underset{x\in I_N}{\sup}\Big|[F_N(x)-F_N(\xi_\alpha)] - [F(x)-F(\xi_\alpha)] \Big| = O(N^{-3/4}\log L) 
		\label{eq:theorem1}
	\end{equation}
	almost surely. Further under Assumption 2.4, we have,
	\begin{equation}
		\Big|[\alpha-F_N(\xi_\alpha)] - (\hat{\xi}^{(P)}_{
			\alpha} - \xi_\alpha)f(\xi_\alpha) \Big| = O(N^{-3/4}\log L), \label{eq:theorem2}
	\end{equation}
	almost surely.
\end{theorem}

\begin{proof}
	Let $\eta_{r,N} = \xi_\alpha + rN^{-3/4}\log L$, where $r = 0, \pm 1, \ldots, \pm b_N$, and $b_N = \lceil N^{1/4} \rceil$. Then for all $x \in J_{r,N} = [\eta_{r,N},\eta_{r+1,N}]$, we have that
	%since both $F_N$ and $F$ are non-decreasing,
	%\begin{align}
	%&[F_N(x)-F_N(\xi_\alpha)] - [F(x)-F(\xi_\alpha)] \nonumber \\
	%&\le F_N(\eta_{r+1,N}) - F_N(\xi_\alpha) - F(\eta_{r,N}) + F(\xi_\alpha) \nonumber \\
	%& = \Big\{[F_N(\eta_{r+1,N})-F_N(\xi_\alpha)] - [F(\eta_{r+1,N})-F(\xi_\alpha)] \Big\} \nonumber \\
	%&+ \{F(\eta_{r+1,N}) - F(\eta_{r,N})\}, \nonumber
	%\end{align}
	%and similarly
	%\[
	%[F_N(x)-F_N(\xi_\alpha)] - [F(x)-F(\xi_\alpha)] 
	%\]
	%\[
	%\ge \Big\{[F_N(\eta_{r+1,N})-F_N(\xi_\alpha)] - %[F(\eta_{r+1,N})-F(\xi_\alpha)] \Big\} 
	%- \{F(\eta_{r+1,N}) - F(\eta_{r,N})\}. 
	%\]
	%Therefore,
	\begin{align*}
		&\underset{x\in I_N}{\sup}\Big|[F_N(x)-F_N(\xi_\alpha)] - [F(x)-F(\xi_\alpha)] \Big| \\\nonumber
		&\le \underset{-b_N\le r \le b_N}{\max}\Big|[F_N(\eta_{r,N})-F_N(\xi_\alpha)] - [F(\eta_{r,N})-F(\xi_\alpha)] \Big| \\\nonumber
		&\quad + \underset{-b_N\le r \le b_N-1}{\max} \big|F(\eta_{r+1,N}) - F(\eta_{r,N}) \big|. \nonumber %\label{eq:theo_part2}
	\end{align*}
	Since $\eta_{r+1,N} - \eta_{r,N} = N^{-3/4}\log L$, 
	%if $F^\prime(x)=f(x)$ continuous in some neighborhood of $\xi_\alpha$ (\textbf{add this as another assumption in section 2}), 
	by the Mean-Value Theorem, 
	\[
	\big|F(\eta_{r+1,N}) - F(\eta_{r,N}) \big| \le \bigg|\underset{x\in J_{r,N}}{\sup}f(x) \bigg|(\eta_{r+1,N} - \eta_{r,N}) = O(N^{-3/4}\log L),
	\]
	and
	$\underset{-b_N\le r \le b_N-1}{\max} \big|F(\eta_{r+1,N}) - F(\eta_{r,N}) \big|=O(N^{-3/4}\log L)$
	almost surely. 
	
	For $r=1,2,\ldots,b_N$, let $U_{ji}^{(r)} = \mathbb{I}(X_{ji}\le \eta_{r,N}) - \mathbb{I}(X_{ji}\le \xi_\alpha)$. Notice that  $U_{ji}^{(r)}$ is 0-1 valued, and such that,
	\[
	F_N(\eta_{r,N})-F_N(\xi_\alpha) = 
	\dfrac{1}{N}\sum_{j=1}^{R}\sum_{i=1}^{L}U_{ji}^{(r)} 
	\]
	and $F(\eta_{r,N})-F(\xi_\alpha) = \mbox{P}(U_{ji}^{(r)}=1) =: \alpha_N^{(r)}$.
	%(\textbf{is $p_N$ defined in the above equation, if so, use triangle equal or =: to define it. We may change the notation to $\alpha_N$ for consistence.})
	
	According to the Mean-Value Theorem and the definition of $b_N$, we get $K_1N^{-3/4}\log L \le p_N^{(r)}\le K_2N^{-1/2}\log L$. By directly applying Lemma~\ref{lemma2},
	%(\textbf{refer to the lemma not the equation})
	we have that, as $L\to \infty$,
	\begin{equation}
		\mbox{P}\left\{\Big|[F_N(\eta_{r,N})-F_N(\xi_\alpha)] - [F(\eta_{r,N})-F(\xi_\alpha)] \Big|>CN^{-3/4}\log L \right\} \le C_2L^{-2}\nonumber
	\end{equation}
	%$\le C_2L^{-2}$ 
	if let $s=2$ in Lemma~\ref{lemma2}.
	
	For $r=-b_N,\ldots,-1$, let $U_{ji}^{(r)} = \mathbb{I}(X_{ji}\le \xi_\alpha) - \mathbb{I}(X_{ji}\le \eta_{r,N})$, and we could derive the same results. According to the Bonferroni Inequality,
	\begin{align}
		&\mbox{P}\Bigg\{\underset{-b_N\le r \le b_N}{\max} \Big|[F_N(\eta_{r,N})-F_N(\xi_\alpha)] - [F(\eta_{r,N})-F(\xi_\alpha)] \Big|
		>CN^{-3/4}\log L \Bigg\} 
		\nonumber \\
		& \le C_2\cdot 2b_N \cdot L^{-2} = O(L^{-3/2}) .\nonumber
	\end{align}
	%(\textbf{emphasize the step you use $R=o(L)$})
	%as $L \to \infty$ and $R=o(L)$. 
	Then, by Borel-Cantelli Lemma \cite{serfling2009approximation},
	%(\textbf{any result not in the undergraduate math textbook, give a reference, it can be some advanced probability textbook}),
	\[
	\underset{-b_N\le r \le b_N}{\max}\Big|[F_N(\eta_{r,N})-F_N(\xi_\alpha)] - [F(\eta_{r,N})-F(\xi_\alpha)] \Big| = O(N^{-3/4}\log L)
	\] 
	%(\textbf{display what is (8) }) 
	holds almost surely. Then Equation~\eqref{eq:theorem1} holds. 
	
	\begin{sloppypar}
		
		We now prove \eqref{eq:theorem2}. 
		%as $L\to \infty$ and $R=o(L)$. Then  Theorem~\ref{theorem1} holds.
		%\begin{comment}
		%\begin{theorem}\label{theorem2}
		%If $f^\prime(x)=\dfrac{d}{dx}f(x)$ is positive and bounded in some neighborhood of $\xi_\alpha$ (\textbf{put this as another assumption, change the statement of this theorem as the theorem 3.1}), then as $L\to \infty$ and $R=o(L)$,
		%\end{theorem}
		%\end{comment}
		Let $k = \lceil N\alpha\rceil$,
		\begin{align}
			&\mbox{P} \left(\hat{\xi}^{(P)}_{\alpha} < \xi_\alpha - N^{-1/2}\log L \right) 
			%& = \mbox{P}\left\{k\ or\ more\ of\ X_{ji} < \xi_\alpha - N^{-1/2}\log L \right\} \nonumber \\
			= \mbox{P}\left\{\sum_{j=1}^{R}\sum_{i=1}^{L}\mathbb{I}\left(X_{ji} \le \xi_\alpha - N^{-1/2}\log L \right) \ge k \right\}\nonumber\\
			& = \mbox{P}\Bigg\{\dfrac{1}{N}\sum_{j=1}^{R}\sum_{i=1}^{L}W_{ji} - F\left(\xi_\alpha - N^{-1/2}\log L \right) 
			\ge \dfrac{k}{N} - F\left(\xi_\alpha - N^{-1/2}\log L \right) \Bigg\}, \nonumber 
		\end{align}
		where $W_{ji} = \mathbb{I}\left(X_{ji} \le \xi_\alpha - N^{-1/2}\log L \right)$ and $\mbox{P}\{W_{ji}=1\} =F\left(\xi_\alpha - N^{-1/2}\log L \right)$. Since as $L\to \infty$, 
		\[
		\dfrac{k}{N} - F\left(\xi_\alpha - N^{-1/2}\log L \right) = f(\xi_\alpha)N^{-1/2}\log L[1+o(1)].
		\]
		From Lemma~\ref{lemma1}, as $L\to \infty$, 
		\[
		\dfrac{1}{N}\sum_{j=1}^{R}\sum_{i=1}^{L}W_{ji} - \mbox{P}\{W_{ji}=1\} \le (\dfrac{2}{t})N^{-1/2}\log L,
		\]
		holds almost surely. Thus, we have that
		\begin{equation}
			\hat{\xi}^{(P)}_{\alpha} \ge \xi_\alpha - N^{-1/2}\log L \label{neq:theo2_part1}
		\end{equation}
		holds almost surely.
		Similarly, 
		\[
		\mbox{P} \left(\hat{\xi}^{(P)}_{\alpha} > \xi_\alpha + N^{-1/2}\log L \right) 
		%& = \mbox{P}\left\{less\ than\ k\ of\ X_{ji} \le \xi_\alpha + N^{-1/2}\log L \right\} \nonumber \\
		= \mbox{P}\left\{\dfrac{1}{N}\sum_{j=1}^{R}\sum_{i=1}^{L} \mathbb{I}\left(X_{ji} \le \xi_\alpha + N^{-1/2}\log L \right) < \dfrac{k}{N} \right\}. 
		\]
		By the monotonicity of $F(\cdot)$, 
		as $L\to \infty$, 
		\[\dfrac{1}{N}\sum_{j=1}^{R}\sum_{i=1}^{L} \mathbb{I}\left(X_{ji} \le \xi_\alpha + N^{-1/2}\log L \right)\to F(\xi_\alpha + N^{-1/2}\log L),\]
		and $\dfrac{k}{N}\to F(\xi_\alpha)$. Thus, we have that,
		\begin{equation}
			\hat{\xi}^{(P)}_{\alpha} \le \xi_\alpha + N^{-1/2}\log L \label{neq:theo2_part2}
		\end{equation} 
		as $L\rightarrow \infty$ almost surely.
		Therefore, under Assumption 2.4, the conclusion holds by setting  $x=\hat{\xi}^{(P)}_{\alpha}$ in 
		\eqref{eq:theorem1}.
		
	\end{sloppypar}
\end{proof}

Notice that Equation~\eqref{eq:theorem2} can be rewritten as,
\begin{equation}
	\hat{\xi}^{(P)}_{\alpha} - \xi_\alpha = \dfrac{\alpha-F_N(\xi_\alpha)}{f(\xi_\alpha)} + O(N^{-3/4}\log L), \label{eq:baha_ad}
\end{equation}
which gives the Bahadur representation of sample quantile of the pooled sample paths.
Now we consider the asymptotic distribution of the estimator $\hat{\xi}^{(P)}_{\alpha}$. Let $F_j(x) = \dfrac{1}{L}\sum_{i=1}^{L}\mathbb{I}(X_{ji}\le x)$ for $j=1,\ldots,R$, and $F_N(x) = \dfrac{1}{R}\sum_{j=1}^{R}F_j(x)$.  
Following the general definition in literature (e.g., \cite{sen1972bahadur}),  
we denote
\begin{equation}
	v^2 = v_{0} + 2\sum_{h=1}^{\infty}v_{h}, \label{eq:v_jh}
\end{equation}
where $v_{h} = \mbox{E}\left[\mathbb{I}(X_{j,1}\le \xi_\alpha)\mathbb{I}(X_{j,1+h}\le \xi_\alpha) \right] - \alpha^2$, which is the same for all replications with $j=1,2,\ldots,R$. 
Under the setting of pooled sample paths in this paper, 
we obtain that
\[
\underset{L\to \infty}{\lim}\left\{N\cdot\mbox{Var}[F_N(\xi_\alpha)] \right\}
=
\underset{L\to \infty}{\lim}\left\{L\cdot\mbox{Var}[F_j(\xi_\alpha)] \right\} 
= v^2. \nonumber 
\]
%(\textbf{we don't need the subindex j for $v^2$, right?})
%We let $\sigma^2 = v_j^2/[f(\xi_\alpha)]^2$, then would have the following

\begin{theorem}\label{theorem3}
	Under Assumptions 2.1--2.3, and {$\sigma^2 := v^2/[f(\xi_\alpha)]^2$, $0<\sigma^2<\infty$},
	%(\textbf{refer to the assumptions})
	\begin{equation}
		\dfrac{N^{1/2}(\hat{\xi}^{(P)}_{\alpha} - \xi_\alpha)}{\sigma} \overset{d}{\to} \mathcal{N}(0,1) \label{eq:theorem3}.
	\end{equation}
	
\end{theorem}

\begin{proof}
	According to Theorem~\ref{theorem1}, 
	%\[
	%\Big|[F_N(\hat{\xi}^{(P)}_{\alpha})-F(\hat{\xi}^{(P)}_{\alpha})] - [F_N(\xi_\alpha) - \alpha] \Big| = O(N^{-3/4}\log L),
	%\]
	%as $L\to \infty$. 
	%Accordingly, 
	we have 
	\begin{equation}
		N^{1/2}[F_N(\hat{\xi}^{(P)}_{\alpha})-F(\hat{\xi}^{(P)}_{\alpha})] \overset{p}{\to} N^{1/2}[F_N(\xi_\alpha) - \alpha] \label{eq:theo3_part1}
	\end{equation}
	as $L\to \infty$.
	By the central limit theorem for $\phi-$mixing variables (see  \cite{sen1972bahadur} for example), 
	$
	L^{1/2}[F_j(\xi_\alpha) - \alpha]/v \overset{d}{\to} \mathcal{N}(0,1)$
	and by the independence between replications,
	\begin{equation}
		\dfrac{N^{1/2}[F_N(\xi_\alpha) - \alpha]}{v} \overset{d}{\to} \mathcal{N}(0,1).
		\label{eq:theo3_part2}
	\end{equation}
	On the other hand, $F_N(\hat{\xi}^{(P)}_{\alpha}) = k/N = \alpha + O(N^{-1}) = F(\xi_\alpha) + O(N^{-1})$. As $L\to \infty$,
	\begin{align}
		&N^{1/2}[F_N(\hat{\xi}^{(P)}_{\alpha})-F(\hat{\xi}^{(P)}_{\alpha})] = N^{1/2}[F(\xi_\alpha)-F(\hat{\xi}^{(P)}_{\alpha})] + O(N^{-1/2}) \nonumber \\
		& = \left[N^{1/2}(\xi_\alpha - \hat{\xi}^{(P)}_{\alpha})f(\xi_\alpha)\right]
		\dfrac{f(\theta\hat{\xi}^{(P)}_{\alpha} + (1-\theta)\xi_\alpha)}{f(\xi_\alpha)} + O(N^{-1/2}) \label{eq:theo3_part3}
	\end{align}
	where $\theta \in [0,1]$. Since $f(x)$ is continuous in some neighborhood of $\xi_\alpha$, $0<f(\xi_\alpha)<\infty$, and from Theorem~\ref{theorem1}, $\big|\xi_\alpha - \hat{\xi}^{(P)}_{\alpha} \big|\le N^{-1/2}\log L$, then as $L\to \infty$,
	$
	f(\theta\hat{\xi}^{(P)}_{\alpha} + (1-\theta)\xi_\alpha)/f(\xi_\alpha) \overset{p}{\to} 1$.
	By applying \eqref{eq:theo3_part1}, \eqref{eq:theo3_part2} and \eqref{eq:theo3_part3}, and the Slutsky's Theorem, \eqref{eq:theorem3} holds.
\end{proof}

\begin{sloppypar}
	According to \eqref{eq:baha_ad} and \eqref{eq:theorem3}, for the sample 
	$\alpha$-quantile $\hat{\xi}^{(P)}_{\alpha}$ given by \eqref{eq:qhat_ad}, we have the following asymptotic bias and variance:
	%when $L\to\infty$ and $R=o(L)$, (\textbf{change to under assumptions ...})
	\begin{equation}
		\mbox{Bias}(\hat{\xi}^{(P)}_{\alpha}) = O(N^{-3/4}\log L), 
		~\mathrm{and}~\mbox{Var}(\hat{\xi}^{(P)}_{\alpha}) = \dfrac{v^2}{N[f(\xi_\alpha)]^2}. \label{eq:bv_ad}
	\end{equation}
	On the other hand, for the classical sample quantile $\hat{\xi}^{(A)}_{\alpha}$ given by \eqref{eq:qhat_cl}, we can 
	directly apply the results from \cite{sen1972bahadur} to obtain 
	\begin{equation}
		\mbox{Bias}(\hat{\xi}^{(A)}_{\alpha})  = O(L^{-3/4}\log L) 
		~\mathrm{and}~
		\mbox{Var}(\hat{\xi}^{(A)}_{\alpha}) = \dfrac{v^2}{N[f(\xi_\alpha)]^2}, \label{eq:bv_cl}
	\end{equation}
	by the mutual independence among different replications. 
	\textit{Asymptotically, $\hat{\xi}^{(P)}_{\alpha}$ achieves the same variance as $\hat{\xi}^{(A)}_{\alpha}$, while
		the bias of  $\hat{\xi}^{(P)}_{\alpha}$ is in a smaller order 
		than the bias of  $\hat{\xi}^{(A)}_{\alpha}$.} We see that the bias of $\hat{\xi}^{(A)}_{\alpha}$ decreases as we increase the run-length, whereas stays the same as we increase the number of replications. Different from the classical quantile estimator, the bias of $\hat{\xi}^{(P)}_{\alpha}$ decreases as we increase either the run-length or the number of replications. Comparing these two estimators using the mean squared error (MSE), the pooled quantile estimator can achieve smaller MSE than the classical average quantile estimator.
\end{sloppypar}

As discussed earlier, %a single sample path could not be generated from parallel processors. However, 
multiple replications could be assigned to parallel processors. Thus, 
given a tight decision time, 
different replications can be allocated to parallel processors to provide a quantile estimator under this time constraint.
With that said, the fixed decision time implies that the run-length is fixed to be $L$, while the number of replications $R$ can be increased depending on how many parallel processors are available. Under this situation, 
the bias of $\hat{\xi}^{(A)}_{\alpha}$ stays at the same order no matter how many processors have been adopted to enlarge the number of replications. Different from $\hat{\xi}^{(A)}_{\alpha}$,  the bias of $\hat{\xi}^{(P)}_{\alpha}$ decreases as we increase the number of parallel processors. 
In Section \ref{sec:num}, we use an empirical example to demonstrate that the pooled quantile estimator outperforms the classical average quantile estimator when the run-length is insufficient due to an urgent deadline.

\section{Numerical Study}\label{sec:num}

In this section, we provide an empirical example to illustrate the performance of the proposed approach 
under the parallel computing setting. We use MSE to demonstrate the performance of different estimators. 
In the following examples, the MSE is computed with 100 micro-replications:
\begin{equation}
	\widehat{\mbox{MSE}}(\hat{\xi}_{\alpha}) =
	\dfrac{1}{100}
	\sum_{m=1}^{100}\left(\hat{\xi}_{\alpha}^{(m)} - \xi_{\alpha}\right)^2
	\nonumber
\end{equation}
where $\hat{\xi}_{\alpha}^{(m)}$ is the estimator for $\alpha$-level quantile from the pooled or the classical average method at the $m$-th micro-replication. Two examples AR(1) and M/M/1 queue are used to generate the dependent sequences. 
We consider that $R$ parallel processors are available to use, and the simulation running in each processor generate one replication of the dependent sequence. Two situations are considered: (1) there is an urgent deadline to provide a quantile estimator, so we only have time to generate dependent sequences with run-length $L=1000$; and (2) the deadline to provide a quantile estimator is not urgent, and we are able to generate dependent sequences with sufficient run-length $L=10000$. \\

\noindent{\bf Example 1: AR(1)} We consider an AR(1) process. For one replication of the dependent sequence, the outputs are given by
\[
X_{i}=\mu+\phi X_{i-1}+\varepsilon_i,
\]
where $\varepsilon_i$ is a white noise process with zero mean and variance $\sigma^2$.  We fix the mean $\mu=0$ and variance $\sigma^2=1$.
The correlation parameter $\phi$ is varying from 0.3, 0.5, to 0.9. We estimate the quantiles $\xi_{\alpha}$ of the outputs with $\alpha=$50\%, 95\%, and compare the performance of the pooled estimator in \eqref{eq:qhat_ad} with the classical average quantile estimator in \eqref{eq:qhat_cl}. 
The results of MSEs under urgent deadline and non-urgent deadline are given in Figures \ref{fg:ar1_urgent} and \ref{fg:ar1_non_urgent}, respectively.

\begin{figure}[!ht]
	\centering
	\includegraphics[scale=0.6]{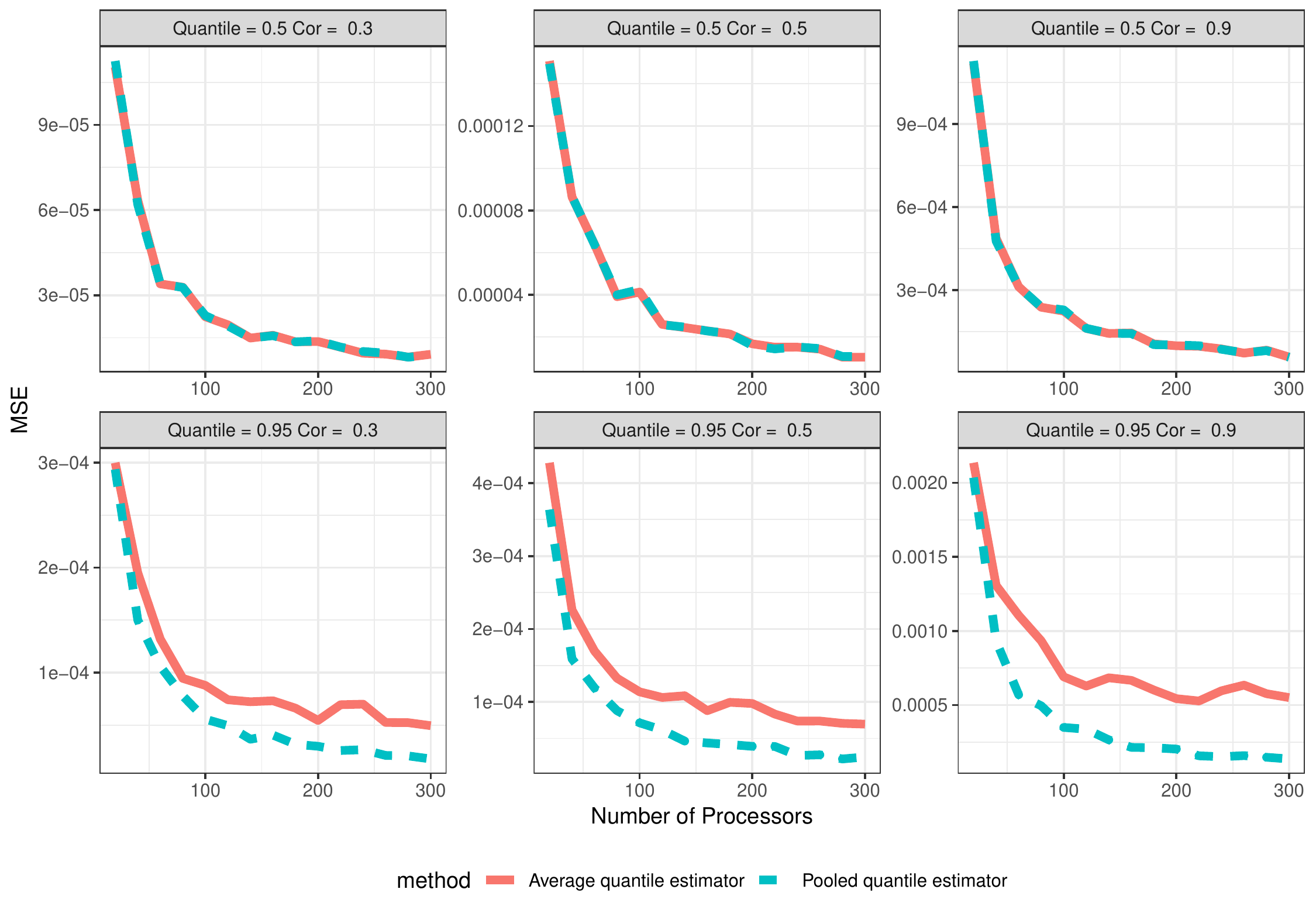}
	\caption{Quantile estimates under urgent deadline ($L=1000$) of the AR(1) example. This figure demonstrates the situation that the experimenter encounters an urgent deadline to provide a quantile estimator, and the time constraint does not allow the run-length on each processor to be sufficient.}\label{fg:ar1_urgent}
\end{figure}

\begin{figure}[!ht]
	\centering
	\includegraphics[scale=0.6]{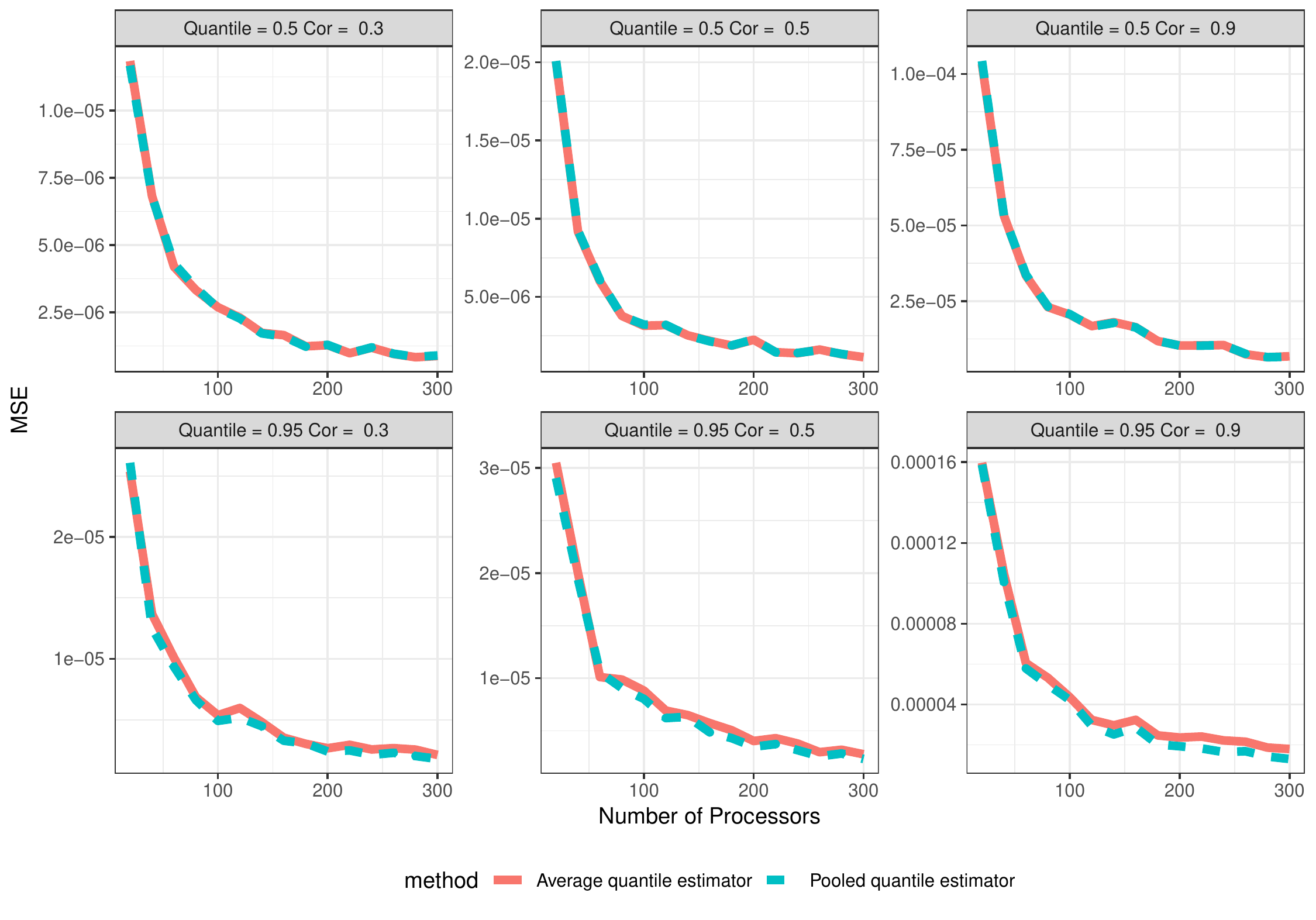}
	\caption{Quantile estimates under non-urgent deadline ($L=10000$) of the AR(1) example. This figure demonstrates that there is no time constraint to generate the quantile estimator, and the run-length can be sufficient to guarantee the accuracy.}\label{fg:ar1_non_urgent}
\end{figure}

\noindent{\bf Example 2: M/M/1 Queue}
We consider the steady-state $M/M/1$ Queueing system. We fix the arrival rate to be 1, and vary the utilization (traffic intensity) to be 0.7 or 0.9. We estimate the quantiles $\xi_{\alpha}$ of time staying in the system  with $\alpha=$50\%, 95\%, and compare the performance of the pooled estimator in \eqref{eq:qhat_ad} with the classical average quantile estimator in \eqref{eq:qhat_cl}. 
The results of MSEs under urgent deadline and non-urgent deadline are given in Figures \ref{fg:mm1_urgent} and \ref{fg:mm1_non_urgent}, respectively. 

%In order to assess the finite sample behaviors, we set the run-length as $L=1000$ and $10000$ representing urgent and non-urgent deadlines, respectively. 

The results shown in Figures  \ref{fg:ar1_urgent}  and  \ref{fg:mm1_urgent}  represent the performances of the quantile estimators under an urgent deadline (i.e., run-length $L=1000$),
whereas the results shown in Figures  \ref{fg:ar1_non_urgent}  and  \ref{fg:mm1_non_urgent}  represent the performances of the quantile estimators under a non urgent deadline (i.e., run-length $L=10000$).
The y-axis represents the estimated MSE and the x-axis is the number of processors $R$, and different scenarios are labeled on top of each sub-figure. 
For the cases representing an urgent deadline, the pooled quantile estimator gives smaller or competitive MSEs compared to the average quantile estimator. This demonstrates the benefits of using the pooled quantile estimator when there is not sufficient time to generate a lengthy sequence.  For the cases representing a non-urgent deadline, we have sufficient time to generate dependent sequences with a relatively large run-length on each processor. As also demonstrated in the theoretical comparison, the MSEs from the pooled estimator significantly outperform the average quantile estimator if the run-length is sufficient.

\begin{figure}[!ht]
	\centering
	\includegraphics[scale=0.7]{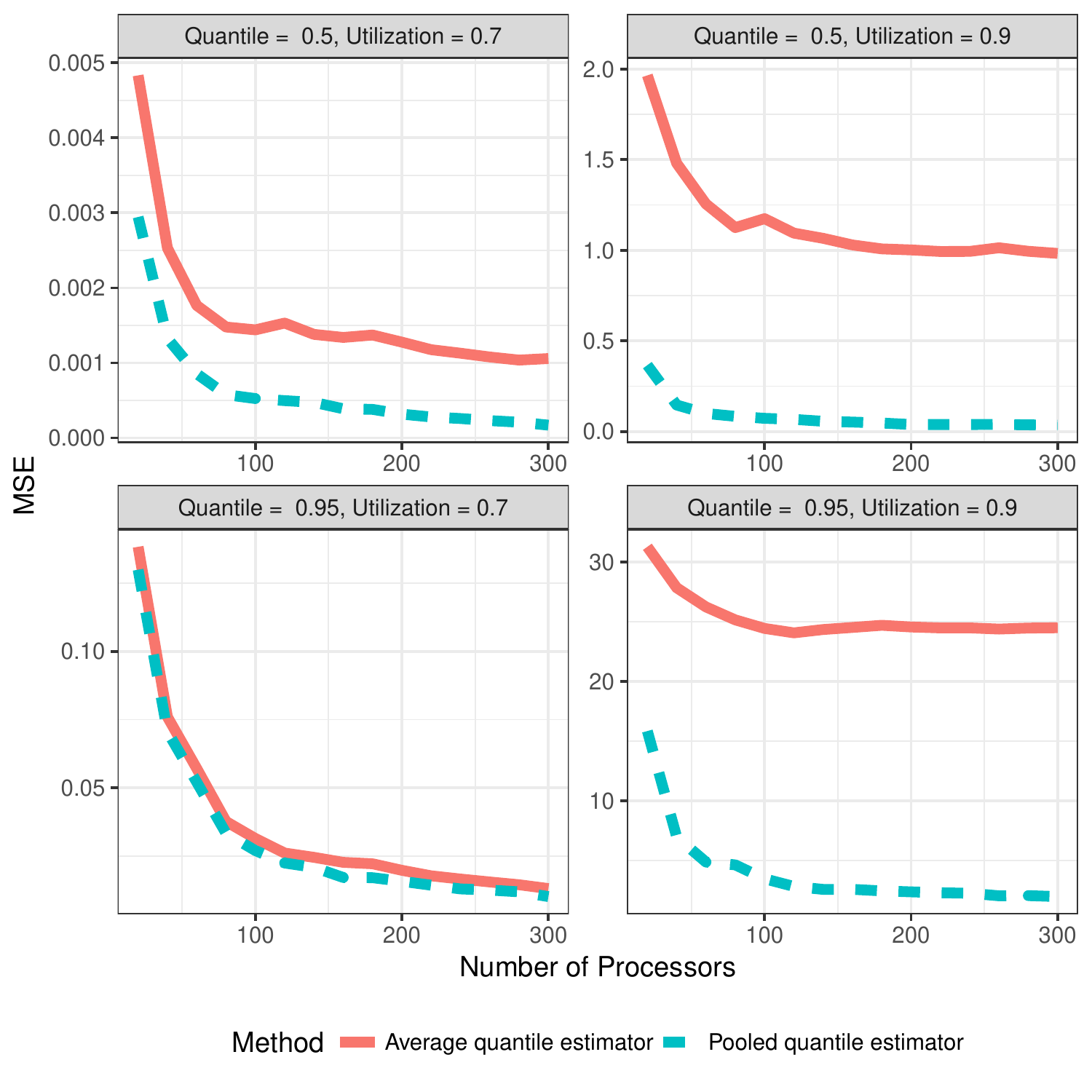}
	\caption{Quantile estimates under urgent deadline ($L=1000$) of the M/M/1 example. This figure demonstrates the situation that the experimenter encounters an urgent deadline to provide a quantile estimator, and the time constraint does not allow the run-length on each processor to be sufficient.}\label{fg:mm1_urgent}
\end{figure}

\begin{figure}[!ht]
	\centering
	\includegraphics[scale=0.7]{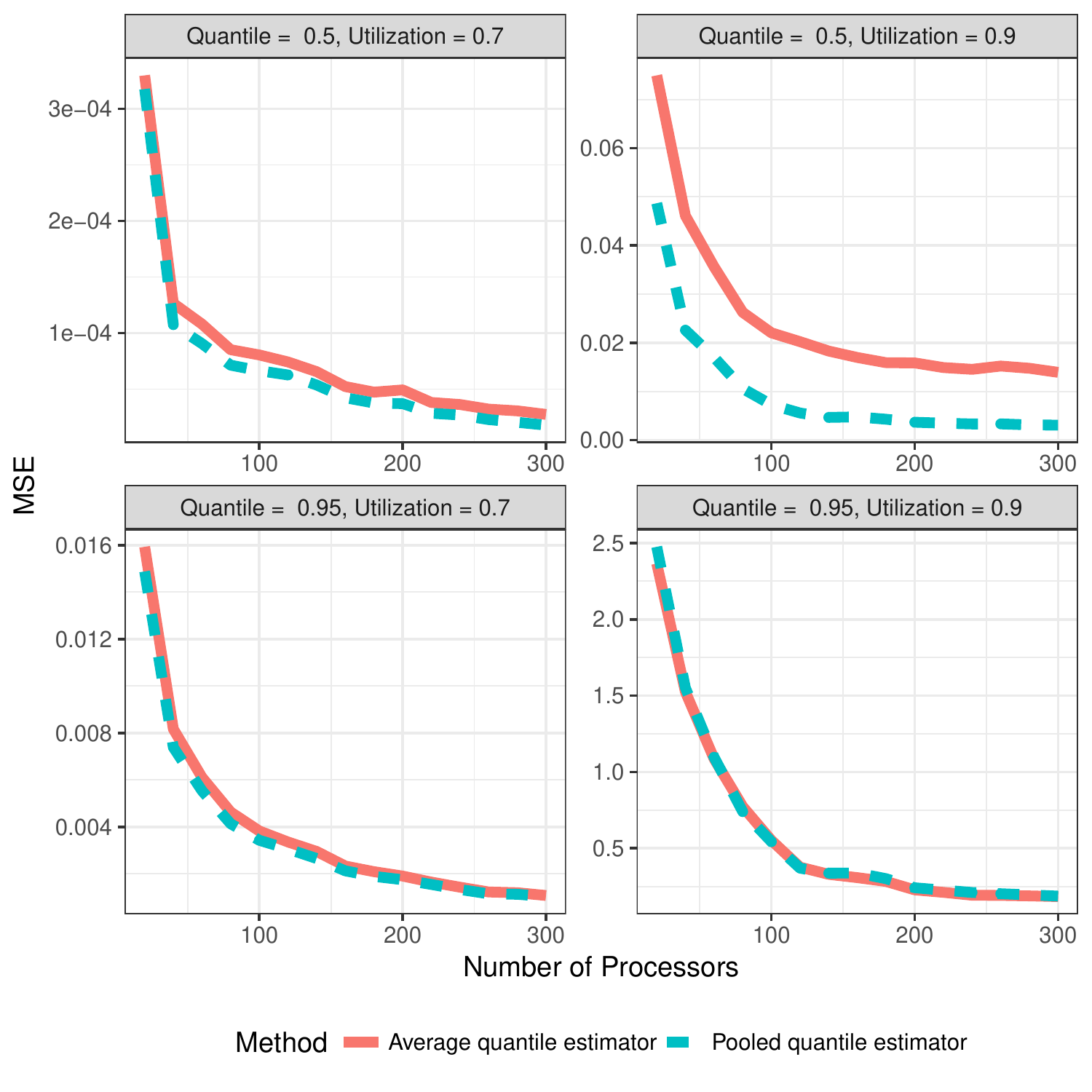}
	\caption{Quantile estimates under non-urgent deadline ($L=10000$) of the M/M/1 example. This figure demonstrates that there is no time constraint to generate the quantile estimator, and the run-length can be sufficient to guarantee the accuracy.}\label{fg:mm1_non_urgent}
\end{figure}

\section{Discussion}
As a summary, we study the pooled quantile estimator, which takes the sample quantile by pooling independently generated sample paths together. We develop the asymptotic representation of the proposed quantile estimator generated from multiple replications of dependent sequences. Compared with the classical average quantile estimator, the pooled quantile estimator demonstrates better asymptotic and finite-sample performance under the context of parallel simulation. 
The pooled quantile estimator can be advanced by combining various existing variance reduction techniques in the literature. Hence, as a promising future direction, the accuracy of the pooled quantile estimator can be further improved by incorporating existing techniques, such as control variates (i.e., \cite{hsu1990control, hesterberg1998control}), importance sampling and stratified sampling
(i.e., \cite{glynn1996importance, glasserman2000variance, sun2010asymptotic}), antithetic variates and Latin hypercube sampling
(i.e., \cite{avramidis1998correlation, jin2003probabilistic}), as well as bias correction techniques (i.e., \cite{matthys2004estimating, gomes2006bias, gomes2007sturdy}). {Also, the $\phi$-mixing condition used in our development can be difficult to check as mentioned in \cite{alexopoulos2019sequest}. It is a promising future direction to extend the asymptotic properties in this paper based on the milder conditions in \cite{wu2005bahadur}. Also, the trade-off between one single replication with multiple replications is an important issue to address for future study. If there is a choice to split one single replication into multiple replications, we may not simply assume that the total simulation cost is $N=R\times L$. The effort spending on the warm-up simulation procedure to achieve steady-state should also be considered. For the case with multiple replications, the warm-up procedure to generate steady-state outputs can not be negligible. It is critically important to further investigate how to balance the trade-off between one single replication with multiple replications.}

% Bibliography
\bibliographystyle{unsrt} 
\bibliography{paperA}
% Sample .bib file with references that match those in
% the 'Specifications Document (V1.5)' as well containing
% 'legacy' bibs and bibs with 'alternate codings'.
% Gerry Murray - March 2012
%
%% History dates
%\received{February 2007}{March 2009}{June 2009}

%\newpage

\appendix

\section{Some Useful Lemmas}
We first consider that $\{Y_{ji}: j=1,2,\ldots,R; i=1,2,\ldots,L\}$ are series of 0-1 valued random variables which satisfies the same mixing condition as $\{X_{ji}\}$ given through \eqref{eq:mixing1} and \eqref{eq:mixing2}, and sharing the same marginal distribution $\mbox{P}(Y_{ji}=1) = 1-\mbox{P}(Y_{ji}=0)=\alpha$. Assume that $S_j=\sum_{i=1}^{L}Y_{ji}$ and $S_N=\sum_{j=1}^{R}S_j$, then we can have the following lemma by extending the results in Section 4 of 
\cite{sen1972bahadur}.

\begin{lemma}\label{lemma1}
	For a positive $t$ ($t<3$) that the $\phi-$mixing condition holds, as $L\to\infty$ and $R=o(L)$,
	\begin{equation}
		S_N\le N\alpha+(\dfrac{2}{t})N^{1/2}\log L,\ w.p.1 \label{eq:lemma1}
	\end{equation}
\end{lemma}

\begin{proof}
	By Markov inequality, we would have,
	\[
	\mbox{P}\left(S_N>N\alpha+(\dfrac{2}{t})N^{1/2}\log L\right) \le
	\underset{h>0}{\inf}\left\{ \exp\left[-hN\alpha-(\dfrac{2}{t})hN^{1/2}\log L\right]
	\mbox{E}[\exp(hS_N)]\right\} \nonumber
	\]
	And we can rewrite
	$S_j = S_j^{(1)} + S_j^{(2)} + \cdots + S_j^{(k_L)}$
	with
	$S_j^{(\ell)} = Y_{j\ell} + Y_{j,{\ell+n_L}} + \cdots + Y_{j,{\ell+m_L^{(\ell)}n_L}}$,
	and choose integer $n_L=\lceil(\dfrac{2}{t})\log L\rceil$, $1\le \ell \le n_L$, and $m_L^{(\ell)}$ be the largest positive integer s.t. $\ell+m_L^{(\ell)}n_L \le L$, notice that $m_L^{(\ell)} \le m_L^{(1)} \le L/n_L - 1$. Then from the independence between replications and inequality between arithmetic and geometric means, we have,
	\begin{align}
		&\mbox{E}[\exp(hS_N)] = \prod_{j=1}^{R}\mbox{E}[\exp(hS_j)]
		= \prod_{j=1}^{R}\mbox{E}\left[\prod_{\ell=1}^{n_L}\exp(hS_j^{(\ell)})\right] \nonumber \\
		&\le \prod_{j=1}^{R}\mbox{E}\left[\left(\dfrac{\sum_{\ell=1}^{n_L}\exp(hS_j^{(\ell)})}{n_L}\right)^{n_L}\right]
		\le \left\{\mbox{E}[\exp(hn_LS_1^{(1)})]\right\}^R \nonumber
	\end{align}
	
	\begin{sloppypar}
		
		According to the $\phi-$mixing condition \eqref{eq:mixing1}, for every $i$,
		$\mbox{P}(Y_{j,i+n_L}=1\mid \mathscr{F}_{-\infty}^{i})\le \alpha + \phi(n_L)$
		and from condition \eqref{eq:mixing2} and the choice of $n_L$, we have,
		\[
		\phi(n_L) = o\left(e^{-tn_L}\right) = o\left[\exp(-2\log L)\right] = o(L^{-2}),\ as\ L\to \infty.
		\]
		Then, for $m = 1,\cdots,m_L^{(1)}$, 
		\begin{align}
			&\mbox{E}\left[\exp(hn_LY_{j,1+mn_L})\mid \mathscr{F}_{-\infty}^{1+(m-1)n_L}\right]
			\nonumber\\
			=& 1 + \mbox{P}\left(Y_{j,1+mn_L}=1\mid \mathscr{F}_{-\infty}^{1+(m-1)n_L}\right)
			[\exp(hn_L)-1] \nonumber \\
			\le& 1 + [\alpha+o(L^{-2})][\exp(hn_L)-1], \label{neq:y1}
		\end{align}
		and $\mbox{E}[\exp(hn_LY_{j1})] = 1 + \alpha[\exp(hn_L)-1]$.
		Applying those recursively yields:
		\begin{align}
			& \mbox{E}[\exp(hn_LS_1^{(1)})] = 
			\mbox{E}\left\{\mbox{E}\left[\exp(hn_L\sum_{m=0}^{m_L^{(1)}}Y_{1,1+mn_L}) \Bigg| \mathscr{F}_{-\infty}^{1+(m-1)n_L}\right]\right\} \nonumber \\
			& \le \mbox{E}\left[\exp(hn_L\sum_{m=0}^{m_L^{(1)}-1}Y_{1,1+mn_L})\right]
			\{1 + [\alpha+o(L^{-2})][\exp(hn_L)-1]\} \nonumber \\
			& \le \cdots \le \{1 + [\alpha+o(L^{-2})][\exp(hn_L)-1]\}^{m_L^{(1)}+1} \nonumber \\
			&\le \exp\left\{\dfrac{L}{n_L}\log\{1 + [\alpha+o(L^{-2})][\exp(hn_L)-1]\} \right\} \nonumber \\
			& \le \exp\left\{\dfrac{tL}{2\log L}\log\{1 + [\alpha+o(L^{-2})][\exp(hn_L)-1]\} \right\} \label{neq:s1}
		\end{align}
		Then, for any $h>0$,
		\begin{align}
			&\mbox{P}\left(S_N>N\alpha+(\dfrac{2}{t})N^{1/2}\log L\right)\le 
			\exp\Bigg\{-hN\alpha-(\dfrac{2}{t})hN^{1/2}\log L \nonumber \\
			&+ \dfrac{tRL}{2\log L}\log\{1 + [\alpha+o(L^{-2})][\exp(hn_L)-1]\} \Bigg\}. \label{neq:p1}
		\end{align}
		By selecting $h=\dfrac{t}{2N^{1/2}\alpha(1-\alpha)}$, we can have,
		\begin{align}
			&\mbox{P}\left(S_N>N\alpha+(\dfrac{2}{t})N^{1/2}\log L\right)\le 
			\exp\Bigg\{-\dfrac{tN^{1/2}}{2(1-\alpha)} - \dfrac{\log L}{\alpha(1-\alpha)} \nonumber \\
			&+ \dfrac{tN}{2\log L}\log\{1 + [\alpha+o(L^{-2})][\exp(\dfrac{\log L}{N^{1/2}\alpha(1-\alpha)})-1]\} \Bigg\} \nonumber \\
			& = \exp\left\{-\dfrac{(1-t/4)\log L}{\alpha(1-\alpha)} + o(1) \right\} 
			= O(L^{-r}) \nonumber
		\end{align}
		where the equality is obtained by applying Taylor's expansion on $\exp(\cdot)$ and $\log(1+\cdot)$, and requires $R=o(L)$. Notice that $\alpha(1-\alpha)\le 1/4$, for $t<3$, we have $r>1$, so \eqref{eq:lemma1} can directly follow Borel-Cantelli Lemma \citep{serfling2009approximation}.
	\end{sloppypar}
\end{proof}

We now consider sequence of 0-1 valued random variables $U_{ji}$ s.t. $U_{ji}=U(X_{ji})$ and $\mbox{P}(U_{ji}=1) = 1-\mbox{P}(U_{ji}=0)=\alpha_N$ for all $j$ and $i$. We can define $S_j$, $S_N$ similar as before by replacing $Y_{ji}$ with $U_{ji}$.

\begin{lemma}\label{lemma2}
	If there exists positive $K_1$ and $K_2$ such that $K_1N^{-3/4}\log L \le \alpha_N\le K_2N^{-1/2}\log L$, for every positive $C$ and $s$, there exists positive $C_s<\infty$ and $L_0(s)$, such that for $L\ge L_0(s)$,
	\begin{equation}
		\mbox{P}\left\{\dfrac{S_N}{N}-\alpha_N > CN^{-3/4}\log L \right\}
		\le C_sL^{-s} \label{eq:lemma2}
	\end{equation}
\end{lemma}

\begin{sloppypar}
	\begin{proof}
		For simplicity, let $C=1$, {and then} $\forall \varepsilon>0$, follow the similar procedure in the proof of Lemma~\ref{lemma1},
		\begin{align}
			&\mbox{P} \left\{ \dfrac{S_N}{N}-\alpha_N > N^{-(3/4-\varepsilon)}\log L \right\} = \mbox{P}\left\{S_N > N\alpha_N + N^{1/4+\varepsilon}\log L \right\} \nonumber \\
			& \le \underset{h>0}{\inf}\left\{ \exp\left[-hN\alpha_N-hN^{1/4+\varepsilon}\log L\right]
			\left(\mbox{E}[\exp(hS_1^{(1)})]\right)^R\right\} \nonumber
		\end{align}
		where $S_j^{(\ell)}$ is also defined similarly as in Lemma~\ref{lemma1}, by replacing $Y_{ji}$ with $U_{ji}$. Again we choose $n_L=\lceil2t^{-1}\log L\rceil$, then \eqref{neq:y1}, \eqref{neq:s1} still hold. Then for any $h>0$,
		\begin{align}
			&\mbox{P}\left\{ \dfrac{S_N}{N}-\alpha_N > N^{-(3/4-\varepsilon)}\log L \right\} \le 
			\exp\Bigg\{ -(\dfrac{2}{t})hN^{1/4+\varepsilon}\log L\nonumber \\
			& -hN\alpha_N+ \dfrac{tRL}{2\log L}\log\{1 + [\alpha_N+o(L^{-2})][\exp(hn_L)-1]\} \Bigg\}.
			\nonumber
		\end{align}
		By selecting $h=\dfrac{s}{N^{1/4+\varepsilon}}$, we can have,
		\begin{align}
			&\mbox{P}\left\{ \dfrac{S_N}{N}-\alpha_N > N^{-(3/4-\varepsilon)}\log L \right\}\le \exp\Bigg\{-stN^{3/4-\varepsilon}\alpha_N - s\log L \nonumber 
		\end{align}
		\begin{align}
			&+ \dfrac{tN}{2\log L}\log\{1 + [\alpha_N+o(L^{-2})][\exp(\dfrac{2s\log L}{tN^{1/4+\varepsilon}})-1]\} \Bigg\} \nonumber \\
			& = \exp\left\{-s\log L + \dfrac{s^2\alpha_N(1-\alpha_N)\log L}{t}N^{1/2- 2\varepsilon} + o(1) \right\} \nonumber
		\end{align}
		where the last equality is obtained by applying Taylor's expansion on $\exp(\cdot)$ and $\log(1+\cdot)$, and requires $R=o(L)$. Since $\alpha_N\le K_2N^{-1/2}\log L$,
		$t^{-1}s^2\alpha_N(1-\alpha_N)\log LN^{1/2- 2\varepsilon} \le O(N^{-2\varepsilon}(\log L)^2)=o(1)$, we have that
		\begin{equation}
			\mbox{P}\left\{ \dfrac{S_N}{N}-\alpha_N > N^{-(3/4-\varepsilon)}\log L \right\}\le \exp\{-s\log L + o(1) \}. \label{neq:p22}
		\end{equation}
		Let $\varepsilon\to 0$, then \eqref{eq:lemma2} directly follows from \eqref{neq:p22}.
	\end{proof}
\end{sloppypar}

\end{document}